\documentclass[5p]{elsarticle}

\usepackage{amssymb,amsmath,amsthm}
\usepackage[vlined]{algorithm2e}

\usepackage{xcolor}

\usepackage{url}

\newcommand{\UWWCM}{\textsc{Unique-Winner WCM}}
\newcommand{\CWWCM}{\textsc{Nonunique-Winner WCM}}

\newcommand{\wnonmanip}{\ensuremath{w^\circ}}
\newcommand{\wmanip}{\ensuremath{w^\bullet}}
\newcommand{\wboth}{\ensuremath{w^{\circ\bullet}}}

\newcommand{\wsummanip}{\ensuremath{W}}

\newcommand{\pinonmanip}{\ensuremath{\pi^\circ}}
\newcommand{\pimanip}{\ensuremath{\pi^\bullet}}
\newcommand{\piboth}{\ensuremath{\pi^{\circ\bullet}}}

\newcommand{\ononmanip}{\ensuremath{\omega^\circ}}
\newcommand{\omanip}{\ensuremath{\omega^\bullet}}
\newcommand{\oboth}{\ensuremath{\omega^{\circ\bullet}}}

 \newtheorem{theorem}{Theorem}
 \newtheorem{corollary}[theorem]{Corollary}
 \newtheorem{lemma}[theorem]{Lemma}
 \newtheorem{proposition}[theorem]{Proposition}

\vfuzz \hfuzz
\makeatletter
\def\ps@pprintTitle{%
	\let\@oddhead\@empty
	\let\@evenhead\@empty
	\def\@oddfoot{}%
	\let\@evenfoot\@oddfoot}
\makeatother

\begin{document}
\begin{frontmatter}
\title{A Note on the Complexity of Manipulating\\  Weighted Schulze Voting\tnoteref{tn,license}}

\tnotetext[tn]{Work done in part while the first author was with University of Konstanz; material is based upon the first author's bachelor's thesis~\cite{mueller-2013}.}

\tnotetext[license]{This document is the accepted manuscript of a work that has been formally published in Information Processing Letters. To access its final edited and published form, visit \url{https://doi.org/10.1016/j.ipl.2020.105989}. \\
\copyright\space2020. This manuscript version is made available under a Creative Commons BY-NC-ND License. To view a copy of the license, visit \url{https://creativecommons.org/licenses/by-nc-nd/4.0/}.}

\author[jm]{Julian M\"uller}
\ead{Julian.Mueller@gess.ethz.ch}

\author[sk]{Sven Kosub}
\ead{Sven.Kosub@uni-konstanz.de}

\address[jm]{Department of Humanities, Social and Political Sciences, ETH Zurich, Weinbergstrasse 109, CH-8092 Zurich, Switzerland}

\address[sk]{Department of Computer and Information Science, University of Konstanz, D-78457 Konstanz, Germany}

\begin{abstract}
We prove that the constructive weighted coalitional manipulation problem for the Schulze voting rule can be solved
in polynomial time for an unbounded number of candidates and an unbounded number of manipulators.
\end{abstract}

\begin{keyword}
	computational social choice; Schulze voting; computational complexity
\end{keyword}

\end{frontmatter}

\section{Introduction}

The Schulze voting rule \cite{schulze-2011} is a recent method for preference aggregation among a group of agents or voters. 
Because of its many desirable mathematical properties, it has attracted a lot of attention in collective decision making and 
has found its way into implementations for various applications. 

This interest also sparked research on the robustness of election outcomes under Schulze voting (margin of victory~\cite{reisch-rothe-schend-2014}) and the rule's resistance to outcome engineering (bribery~\cite{hemaspaandra-lavaee-menton-2016,parkes-xia-2012}, control~\cite{hemaspaandra-lavaee-menton-2016,menton-singh-2013,parkes-xia-2012} and voting manipulation~\cite{gaspers-kalinowski-narodytska-walsh-2013,hemaspaandra-lavaee-menton-2016,parkes-xia-2012}). In particular, several studies have suggested that the Schulze voting rule is vulnerable to voting manipulation 
({\em aka} strategic voting). 
In the light of the famous Gibbard-Satterthwaite theorem \cite{gibbard-1973,satterthwaite-1975} (and its extension by the 
Duggan-Schwartz theorem \cite{duggan-schwartz-2000} to irresolute voting rules), which in principle excludes non-manipulatibility 
in most cases, voting manipulation could be made prohibitive by using computational hardness 
(cf., e.g., \cite{bartholdi-tovey-trick-1989,conitzer-sandholm-lang-2007,conitzer-walsh-2016}).
However, the manipulation problem for Schulze voting has been shown to be solvable in polynomial time 
for many computational scenarios.

Parkes and Xia~\cite{parkes-xia-2012} showed that, for Schulze voting, constructive unweighted coalitional manipulation (UCM) for a 
single manipulator and destructive UCM for an arbitrary number of manipulators are solvable in polynomial time. 
Building on this work, Gaspers {\em et al.} \cite{gaspers-kalinowski-narodytska-walsh-2013} found that 
constructive UCM is polynomial-time computable and established that constructive weighted coalitional manipulation 
(WCM) is fixed-parameter tractable (in the number of candidates) in the nonunique-winner model. 
In addition, Hemaspaandra, Lavaee and Menton \cite{hemaspaandra-lavaee-menton-2016} proved that constructive WCM 
in the unique-winner model and destructive WCM in both models are also fixed-parameter tractable.

We improve on these results and give an algorithm that solves constructive WCM in polynomial time even for 
an unbounded number of candidates.
As the Parkes-Xia algorithm \cite{parkes-xia-2012} for destructive UCM in the unique-winner model is easily adaptable to destructive WCM in both the unique 
and the nonunique-winner model,\footnote{Since all manipulators vote in the same way in the Parkes-Xia algorithm, weights actually do not matter for
the construction of manipulators' voting profiles. Furthermore, the number of manipulators and non-manipulators has no influence on the runtime of 
the core algorithm except for computing the weighted majority graph and outputting the voting profiles (for the relevant notions, see Sect.~2). 
Thus, the algorithm runs in polynomial time on weighted~instances.} all open problems identified in 
\cite{hemaspaandra-lavaee-menton-2016} regarding the complexity of WCM for Schulze voting are thus settled: they are all 
polynomial-time computable.

\section{The Schulze voting rule}

{\em Voting.} 
Let $X=\{x_1,\dots, x_m\}$ be a finite, non-empty set of $m$ alternatives or candidates, and let
$V=\{v_1,\dots, v_n\}$ be a finite, non-empty multi-set of $n$ agents or voters.
We assume that each voter $v_i$ has a weight $w_i\in \mathbb{N}\setminus \{0\}$; 
e.g., $v_i$ could stand representatively for a coalition of $w_i$ unanimous voters.
Each voter $v_i\in V$ imposes a total order on the set of candidates which expresses $v_i$'s 
preferences and which is called $v_i$'s {\em vote};  
more specifically, a vote of voter $v_i\in V$ is given by a bijective function 
$\pi_i: X \to \{1,\dots, m\}$, where $\pi_i(x)>\pi_i(y)$ means that $v_i$ strictly prefers $x$ to $y$.\footnote{Note that incomparable candidates are excluded. Note also that we reverse the usual preference order (with best alternative at rank~1), which is more appropriate for score-based representations of rankings in line with algorithmic constructions in this article.}
A weighted {\em voting profile} is a sequence $(\pi,w)=(\pi_1,\dots,\pi_n; w_1,\dots, w_n)$ containing 
a vote (of weight $w_i$) from each voter $v_i$.

A voting rule is a mapping $F : (\pi_1,\dots, \pi_n; w_1,\dots, w_n) \linebreak \mapsto Y \in {\cal P}(X)$, i.e.,
$F$ assigns a set of candidates to each possible weighted voting profile $(\pi_1,\dots, \pi_n;w_1,\dots, w_n)$.
If $F(\pi_1,\dots,\pi_n; w_1,\dots, w_n)$ is a singleton set then the winning candidate is called 
{\em unique} winner of the election under voting rule $F$.

\bigskip

{\em Weighted majority graph.}
The Schulze voting rule is based on the method of pairwise comparisons. 
Given a weighted voting profile $(\pi,w)=(\pi_1,\dots, \pi_n; w_1,\dots,w_n)$ and distinct 
candidates $x,y\in X$, let $\omega_{(\pi,w)}(x,y)$ denote the weighted number of voters that strictly prefer $x$ 
to $y$ minus the weighted number of voters that strictly prefer $y$ to $x$, i.e.,
\[\omega_{(\pi,w)}(x,y) =_{\rm def} \sum_{\pi_i(x)>\pi_i(y)} w_i  ~~- \sum_{\pi_i(y)>\pi_i(x)} w_i.\]
So, if $\omega_{(\pi,w)}(x,y)>0$ then the majority of voters prefers $x$ to $y$.

The outcomes of the pairwise comparisons are collected as edge annotations to the complete directed graph 
$K=(X,E)$ on vertex set $X$, i.e, the edge set of $K$ is $E=X\times X \setminus \{~(x,x)~|~x\in X~\}$.
In particular, for weight function $\omega_{(\pi,w)} : E\to \mathbb{Z} : (x,y) \mapsto \omega_{(\pi,w)}(x,y)$, 
the skew-symmetric graph $(K,\omega_{(\pi,w)})$ is the {\em weighted majority graph} for the weighted voting 
profile $(\pi,w)$.
In the forthcoming, we also consider weight functions for $K$ that do not necessarily correspond to voting profiles.

\bigskip

{\em Schulze method.} The Schulze voting rule evaluates the weighted majority graph in a specific path-based way.
For our purposes, it is useful to introduce the relevant notions for a weighted graph $(K,\omega)$ 
equipped with an arbitrary weight function $\omega:E\to\mathbb{Z}$.
For each path $p=(y_1,\dots, y_\ell)$ in $(K,\omega)$, the {\em strength $s[\omega](p)$~of path $p$} is defined 
to be the minimum weight of two candidates consecutively occurring on $p$, i.e., 
\[s[\omega](p)=_{\rm def} \min\  \{~\omega(y_{i}, y_{i+1})~|~i\in \{1,\dots, \ell-1\}~\}.\]
For distinct candidates $x$ and $y$, the {\em path strength $S[\omega](x,y)$} is defined to be the maximum
strength of an $(x,y)$-path in $(K,\omega)$, i.e.,
\begin{multline*}S[\omega](x,y)=_{\rm def} \\ \max\ \{~s[\omega](p)~|~\textrm{$p$ is an $(x,y)$-path in $(K,\omega)$}~\}.
\end{multline*}
Finally, a candidate $x\in X$ is a {\em Schulze winner} for weighted voting profile $(\pi,w)$
if and only if $S[\omega_{(\pi,w)}](x,y)\ge S[\omega_{(\pi,w)}](y,x)$ for all candidates $y\in X\setminus\{x\}$.

The following proposition summarizes both that the Schulze voting rule always selects at least one winner and that deciding whether there is 
a unique Schulze winner is also easily possible.

\begin{proposition} \label{prop:schulzewinner}
Let $(\pi,w)$ be any weighted voting profile for candidate set $X$ and voter set $V$.
\begin{enumerate}
\item {\em \cite{schulze-2011}} There always exists a Schulze winner for the weighted voting profile $(\pi,w)$.\looseness=1
\item {\em \cite{parkes-xia-2012}} A candidate $x\in X$ is a unique Schulze winner for weighted voting profile $(\pi,w)$
	if and only if $S[\omega_{(\pi,w)}](x,y)> S[\omega_{(\pi,w)}](y,x)$ for all $y\in X\setminus \{x\}$.
\end{enumerate}
\end{proposition}

\section{Unique-Winner Constructive Weighted Coalitional Manipulation}\label{sect:uwinner}

Constructive voting manipulation consists of fixing the voting behavior of a coalition of manipulators in order to make a certain candidate 
the winner of the election. More specifically, we consider the following computational problem for finding a unique winner with 
respect to the Schulze voting rule (the nonunique-winner version will be discussed in Sect.~\ref{sect:cowinner}):

\bigskip

\noindent
\begin{tabular}{p{.15\linewidth}p{.77\linewidth}}
  \textsl{Problem:}  	& $\UWWCM$ \\
  \textsl{Instance:}    	& A candidate set $X$, weighted voting profile $(\pi_1,\dots, \pi_k; w_1,\dots, w_k)$ of non-manipulators,
  						  weights $(w_{k+1}, \dots, w_n)$ of manipulators, and a preferred candidate $c\in X$ \\
  \textsl{Task:}   		& Find votes $\pi_{k+1},\dots, \pi_n$ such that $c$ is a unique Schulze winner for the weighted voting 
  						  profile $(\pi_1,\dots,\pi_n; w_1,\dots, w_n)$, or indicate non-existence				  
\end{tabular}

\bigskip

The following theorem is the key to a polynomial algorithm for $\UWWCM$.

\noindent

\begin{theorem}\label{thm:uwwcm}
Let $(X,(\pi_1,\dots, \pi_k;w_1,\dots,w_k),$\linebreak$ (w_{k+1},\dots, w_n), c)$
be a \UWWCM\ instance.
There exists a (finite) function $U:X\to\mathbb{Z}$ which can be computed in polynomial time such that the following two statements are equivalent:
\begin{enumerate}
\item There exist votes $(\pi_{k+1},\dots, \pi_n)$ such that $c$ is a unique Schulze winner for the weighted voting profile 
	$(\pi_1, \dots, \pi_n; w_1, \dots, w_n)$.
\item $U(x)> \omega_{(\pi_1,\dots, \pi_k;w_1,\dots, w_k)}(x,c)-(w_{k+1}+\dots+w_n)$ for all $x\in X\setminus \{c\}$.
\end{enumerate}
\end{theorem}

	The idea behind Theorem \ref{thm:uwwcm} is to construct upper bounds $U(x)$ on the strengths of backward paths from alternative candidates $x$ to the preferred candidate $c$ in any successful manipulation scenario. Consequently, if the second condition of the theorem does not hold true at one candidate $x$, then $c$ will never beat some other candidate due to a backward path that traverses $x$. On the other hand, we will see that such bounds $U(x)$ can be constructed sufficiently small to guarantee the existence of a successful manipulation if the second condition is satisfied at each alternative candidate.

In Lemma \ref{lemma:algoboundsruntime}, we will see that such a function $U$ can be computed in (non-optimal) time $O(m^5)$
(for an $O(m^3)$~bound, see \cite{mueller-2013}), given the weighted majority graph.
Together with checking the second condition of Theorem \ref{thm:uwwcm} for $m-1$ candidates, 
this gives an algorithm whose running time is dominated by the time needed for computing  
function~$U$. 
Note that it takes $O(nm^2)$ steps to compute the weighted majority graph in the first place.

The proof of Theorem \ref{thm:uwwcm} is split into two parts, each of which is covered by its own subsection. 
Subsection~\ref{subsect:necessity} (necessity) describes an algorithm for computing function~$U$ and proves its correctness 
given that there exist manipulators' votes to make $c$ a unique Schulze winner. 
Subsection~\ref{subsect:sufficiency} (sufficiency) describes an algorithm for using function~$U$ to determine the manipulators' 
votes, given that the second condition of Theorem \ref{thm:uwwcm} is satisfied.
Formally, Theorem~\ref{thm:uwwcm} follows from Lemma~\ref{lemma:algoboundsruntime} and Corollary~\ref{cor:sufficiency}.

\subsection{Computing function $U$}\label{subsect:necessity}

{\em Notation.}
We use the following notation for successful $\UWWCM$ instances: 
$(\pinonmanip, \wnonmanip)$ denotes the weighted voting profile of the non-manipulators, i.e., $\pinonmanip=(\pi_1,\dots,\pi_k)$ and $\wnonmanip=(w_1,\dots,w_k)$;
$(\pimanip, \wmanip)$ denotes a successfully manipulating weighted voting profile of the manipulators, i.e., $\pimanip=(\pi_{k+1},\dots,\pi_n)$ and $\wmanip=(w_{k+1},\dots,w_n)$ such that the first condition of Theorem \ref{thm:uwwcm} is satisfied;  
$(\piboth, \wboth)$ denotes a successfully manipulating weighted voting profile of non-manipulators and manipulators, i.e., $(\piboth, \wboth) = (\pinonmanip\pimanip; \wnonmanip\wmanip)$.
Likewise, we use $\ononmanip,\omanip$, and $\oboth$ to denote the edge weight function of the weighted majority graph with respect to the corresponding scenarios.
Let $W$ denote the total weight of the manipulators, i.e., $W=_{\rm def} w_{k+1}+\dots+w_n$.

\bigskip

\noindent
{\em Algorithm.}
The following algorithm (involving two rules we fix later) contains a schematic description of how to compute function $U$: 

\begin{algorithm}
\caption{Algorithm for computing $U$ \label{algorithm:bounds}}
\KwIn{Graph $(K,\ononmanip)=(X,E,\ononmanip)$, candidate $c\in X$, total weight $\wsummanip$ of the manipulators}
\KwOut{(finite) function $U:X\to \mathbb{Z}$}
\Begin{
   	$U(c) \leftarrow \infty$\footnotemark\;
     \ForEach{$x \in X \setminus \{c\}$}{
          $U(x) \leftarrow$ \\ \hfill $\max\ \{\ \ononmanip(y,z)\ |\ y,z\in X, y\neq z\ \} + \wsummanip$\;
     }
     \While{Rule~1 or Rule~2 is applicable to $U$}{
          apply the rule\;
     }
    \Return $U$\;
}
\end{algorithm}
\footnotetext{The $\infty$ symbol serves as a placeholder for a sufficiently large number that acts neutrally in all relevant minimum operations and exceeds $U(x)$ at all alternative candidates $x \in X \setminus \{c\}$; it thus behaves like positive infinity compared to all other values that appear in presented proofs and algorithms. To be exact, $\infty$ may denote any number that strictly exceeds $\ononmanip(x,y) + \wsummanip$ for all~$x, y \in\nolinebreak X, x\neq y$.}

\noindent
While computing $U$, Algorithm~\ref{algorithm:bounds} maintains the following invariant (in accordance with the intended second condition of Theorem \ref{thm:uwwcm}):
\begin{quote}
For any successfully manipulating weighted profile $(\piboth, \wboth)$, we have $U(x) > S[\oboth](x, c)$ for all candidates 
$x \in X \setminus \{c\}$.
\end{quote}
We observe that the initialization of $U(x)$ satisfies the invariant, as $S[\oboth](c, x)$ cannot exceed $U(x)$ for any $x \in X$.
In order to decrease the initial bounds $U(x)$ to a sufficient degree, we use the following rules maintaining the invariant
in Algorithm 1:\footnote{Gaspers {\em et al.} \cite{gaspers-kalinowski-narodytska-walsh-2013}, in their 
algorithm which has the same structure as ours, used three rules. The reason that Algorithm 1 runs in polynomial time lies in the 
difference between their Rule 2, which reduces the values of the bounds by 2 per application, and our Rule 1, which is more efficient in respect thereof. 
Rule 3 of Gaspers {\em et al.} corresponds to our Rule 2. }
\begin{itemize}
\item {\em Rule 1.}
If there is a candidate $x\not=c$ such that $S[\omega'](c,x)< U(x)$ when using weight function 
$\omega':\nolinebreak E\to\mathbb{Z}: (y,z) \mapsto\min\ \{\ \ononmanip(y,z) + \wsummanip, U(z)\ \}$, set $U(x) \leftarrow S[\omega'](c,x)$.
\item {\em Rule 2.}
If there are candidates $x,y \in X \setminus \{c\}$ such that $U(y) < U(x)$ and $\ononmanip(y,x) - \wsummanip \ge U(y)$, 
set $U(x) \leftarrow U(y)$.
\end{itemize}

\begin{lemma} \label{lemma:rule1correctness} Rule~1 maintains the invariant. \end{lemma}
\begin{proof}
Suppose that the invariant is violated for the first time after the rule has been applied, i.e.,\linebreak $S[\omega'](c, x) = U(x)\le S[\oboth](x,c)$. 
Since $c$ is the unique winner for profile $(\piboth, \wboth)$, there exists a $(c,x)$-path $p$ such that $s[\oboth](p) > S[\oboth](x,c)$. 
Thus, $p$ must satisfy $s[\ononmanip](p)+\wsummanip \ge s[\oboth](p) > S[\omega'](c,x) \ge s[\omega'](p)$, so there must be a candidate 
$y$ on path $p$ with $U(y) \le S[\omega'](c,x)$. 
Let $p'$ denote the $(y,x)$-subpath of~$p$. Then, $S[\oboth](y,c)\ge \min \{ s[\oboth](p'), S[\oboth](x,c) \} \ge S[\omega'](c,x) \ge U(y)$. But this would mean that the invariant had already been violated before the rule had been applied.
\end{proof}

\begin{lemma} \label{lemma:rule2correctness} Rule~2 maintains the invariant. \end{lemma}

\begin{proof}
Suppose the invariant is violated for the first time after the rule has been applied to candidates $x$ and $y$, i.e., $S[\oboth](x,c)\ge U(y)$ ($=U(x)$ after the update). 
Then, $S[\oboth](y,c)\ge \min \{ \wnonmanip(y,x) - \wsummanip, S[\oboth](x,c) \} \ge U(y)$. This is a contradiction, as it implies that the invariant had already been violated before the rule was even applied.
\end{proof}

\begin{lemma}\label{lemma:algoboundsruntime}
Algorithm~\ref{algorithm:bounds} computes a function $U: X \rightarrow \mathbb{Z}$ in $O(m^5)$ time for which 
it holds that if the manipulators can vote such that $c$ is the unique Schulze winner, then 
$U(x) > \ononmanip(x, c) - \wsummanip$ for all $x \in X \setminus \{c\}$.
\end{lemma}

\begin{proof} The runtime bound follows from two observations:
		\begin{itemize}
			\item We can check if and where a rule can be applied and apply it in $O(m^2)$ time.
			\item Whenever a rule is applied at a candidate $x$, $U(x)$ must strictly decrease. As $U(x)$ can only assume values from the set $\{~ \ononmanip(y, z) + \wsummanip~|~y, z \in X~\}$, such a decrease can only happen at most $O(m^2)$ times at one specific candidate, and thus at most $O(m^3)$ times in total.
		\end{itemize}
	The claimed inequality is a consequence of the invariant, since $S[\oboth](x,c) \geq S[\ononmanip](x,c) - \wsummanip \geq \ononmanip(x, c) - \wsummanip$.
\end{proof}

\subsection{Computing the votes of the manipulators} \label{subsect:sufficiency}

{\em Algorithm.} We have seen by now that we can use function $U$ computed by Algorithm~\ref{algorithm:bounds} to formulate a necessary condition for when the manipulators can be successful. We will see now that this necessary condition is also sufficient, as we can also utilize function $U$ from Algorithm~\ref{algorithm:bounds} to construct votes for the manipulators that make $c$ the unique Schulze winner. Algorithm~\ref{algorithm:constructvote} shows how to construct these \nolinebreak votes in $O(m^2)$ time.

\begin{algorithm}
	\caption{Computation of the votes of the manipulators \label{algorithm:constructvote}}
	\KwIn{Graph $(K,\ononmanip)=(X,E,\ononmanip)$,
		candidate $c \in X$, total weight $\wsummanip$ of the manipulators, function $U$ computed by Algorithm~\ref{algorithm:bounds}}
	\KwOut{A vote $\zeta:X\rightarrow \{1,\dots,m\}$ that can be voted by all the manipulators to achieve a successful manipulation, if the election can be manipulated}
	\Begin{
		{Construct directed graph $G=(X, E')$ such that, for $x\not= y$, $(x,y) \in E' ~\iff~$ \\ \quad\quad\quad\quad\quad $\min\ \{\ U(x), \ononmanip(x,y) + \wsummanip\ \} \geq U(y)$}\;
		Compute a spanning arborescence $T \subseteq G$ with root $c$\;
		Compute vote $\zeta : X \rightarrow \{1, \dots, m\}$ by topologically sorting $T$ such that $(x, y) \in E(T) \vee U(x) > U(y) 
		~\implies~ \zeta(x) > \zeta(y)$ for all $x,y \in X$\;
		\Return $\zeta$\;
	}
\end{algorithm}

\begin{lemma} \label{lemma:pathinG}
	For all candidates $x\in X$, there exists a $(c,x)$-path in the graph $G$ as defined in Algorithm~\ref{algorithm:constructvote}.
\end{lemma}

\begin{proof}

	Suppose the claim does not hold, and choose $x$ without a $(c, x)$-path in $G$ such that $U(x)$ is maximum among all candidates without such a path. 
	Since Rule~1 is inapplicable to $U$ at $x$, there exists a $(c,x)$-path $p = (y_1,\dots, y_\ell)$ in $K$ 
	with $\min\{ \ononmanip(y_{i-1}, y_i) + \wsummanip, U(y_i) \} \geq U(x)$ for all $i \in \{2,\dots, \ell \}$ (as follows from the definition of $\omega'$ in Rule~1). 
	Since $\infty = U(c) > U(x)$, there is a candidate $y_i$ with maximal index $i$ such that $U(y_i) > U(x)$. 
	Let $p'$ be the $(y_i,x)$-subpath of $p$. 
	This $(y_i,x)$-path $p'$ also exists in~$G$, since $U(y_{i+1})=\dots=U(y_\ell)=U(x)$, and there is also a $(c,y_i)$-path in $G$ because $U(x)$ is maximum among the candidates not reachable from $c$ in $G$.
	Hence, there exists a $(c,x)$-path in $G$. This is a contradiction.
\end{proof}

Thus, there exists a spanning arborescence $T$ of $G$ rooted at $c$, i.e., a directed spanning tree of $G$ with root $c$ such that there is a unique directed path from the root $c$ to any other candidate in the tree. 
From this arborescence, we obtain a vote $\zeta$ such that  $(x,y) \in E(T)$ or $U(x) > U(y)$ implies $\zeta(x) > \zeta(y)$. Since $G$ and thus $T$ do not contain any edge $(x, y) \in E'$ with $U(x) < U(y)$ as a direct consequence of $G$'s definition, such a vote exists and can be constructed through topological sorting. 
We choose all manipulators to vote $\zeta$, so we set 
$\pimanip =_{\rm def} (\zeta, \dots, \zeta)$ and $(\piboth, \wboth) =_{\rm def} (\pi_1, \dots, \pi_k, \zeta, \dots, \zeta; w_1, \dots, w_n)$.\footnote{This construction implies that whenever the manipulators can make their preferred candidate the unique winner, then all of them can vote the same way to achieve this result.}

\begin{lemma} \label{lemma:pathtolowerbound}
	$S[\oboth](c,x)\geq U(x)$ for all candidates $x\in X \setminus \{c\}$.
\end{lemma}

\begin{proof}
	Let $p$ denote a $(c,x)$-path $p = (y_1, \dots, y_\ell)$ in~$T$. 
	All edges in $T$ are also edges in $G$, so we have that $\min \{ U(y_{i-1}), \ononmanip(y_{i-1}, y_i) + \wsummanip \} \geq U(y_i)$ for all~$i \in \nolinebreak \{2, \dots, \ell\}$.  By a simple inductive argument and since the vote $\zeta$ is consistent with all edges in~$T$, we obtain $\oboth(y_{i-1}, y_i) = \ononmanip(y_{i-1}, y_i) + W \ge U(x)$ for all~$i \in\nolinebreak \{2, \dots, \ell\}$. We conclude that $S[\oboth](c, x) \ge s[\oboth](p) \ge U(x)$, which proves the claim.
\end{proof}

\begin{lemma} \label{lemma:pathfromupperbound}
	If $U(x) > \ononmanip(x,c) - W$ for all $x\in X \setminus \{c\}$, then $U(x) > S[\oboth](x,c)$ for all $x\in X \setminus \{c\}$.
\end{lemma}

\begin{proof}
	Let $p$ denote any $(x,c)$-path. We consider two cases:
	\begin{itemize}
		\item Assume that there is a candidate $y\neq c$ on $p$ with $U(y) > U(x)$. 
		Then, choose the $y$ closest to $x$ on $p$ among these candidates, and let~$z$ be the candidate 
		immediately preceding $y$ on~$p$. Because $U(y) > \nolinebreak U(x) \ge\nolinebreak U(z)$, the manipulators vote $\zeta(y) > \nolinebreak \zeta(z)$ and thus the vote count is $\oboth(z,y) =\nolinebreak \ononmanip(z,y) -\nolinebreak \wsummanip$.
		As Rule~2 is not applicable to $U$, we have $\ononmanip(z,y) - \wsummanip < U(z)$, and hence obtain $s[\oboth](p) \leq \oboth(z,y) = \ononmanip(z,y) - \wsummanip < U(z) \leq U(x)$.
		\item Otherwise, it holds for all candidates $y\neq c$ on $p$ that $U(y) \leq U(x)$. 
		Let $z$ be the candidate preceding $c$ on~$p$. 
		Since $\infty = U(c)> U(z)$, the manipulators vote $\zeta(c) > \zeta(z)$ and votes count to $\oboth(z,c) = \ononmanip(z,c) -\nolinebreak \wsummanip$. We conclude that $s[\oboth](p)\le \oboth(z,c) = \ononmanip(z,c) - \wsummanip < U(z) \leq U(x)$ due to the premise of Lemma~\ref{lemma:pathfromupperbound}.
	\end{itemize}
	The claim follows.
\end{proof}

\begin{corollary}\label{cor:sufficiency}
	If $U(x) > \ononmanip(x,c) - W$ for all $x\in X \setminus \{c\}$, then  $c$ is a unique Schulze winner for weighted voting profile $(\pi_1, \dots, \pi_k, \zeta, \dots, \zeta; w_1, \dots, w_n)$.
\end{corollary}

\begin{proof}
	We conclude from the two preceding lemmas that $S[\oboth](c, x) \geq U(x) > S[\oboth](x, c)$ for all $x \in X \setminus \{c\}$. Hence, $c$ is a unique Schulze winner by Proposition~\ref{prop:schulzewinner}.
\end{proof}

\section{Nonunique-Winner Constructive Weighted Coalitional Manipulation}\label{sect:cowinner}

We briefly discuss the nonunique-winner variant of the constructive weighted coalitional manipulation problem, i.e., checking whether
candidate $c$ can be made a Schulze winner (but not necessarily the only one) with respect to the same voting profile.

\bigskip

\noindent
\begin{tabular}{p{.14\linewidth}p{.78\linewidth}}
	\textsl{Problem:}  	& $\CWWCM$ \\
	\textsl{Instance:}    	& A candidate set $X$, weighted voting profile $(\pi_1,\dots, \pi_k; w_1,\dots, w_k)$ of non-manipulators,
	weights $(w_{k+1}, \dots, w_n)$ of manipulators, and a preferred candidate $c\in X$ \\
	\textsl{Task:}   		& Find votes $\pi_{k+1},\dots, \pi_n$ such that $c$ is a Schulze winner for the weighted voting 
	profile $(\pi_1,\dots,\pi_n; w_1,\dots, w_n)$, or indicate non-existence
\end{tabular}

\bigskip

Following the arguments in Sect.~\ref{sect:uwinner} literally, the algorithms for $\UWWCM$ can be easily adapted to solve $\CWWCM$.
A modification is only required for Algorithm 1 computing $U$. 
In contrast to the unique winner problem, we allow the bounds in the invariant to be non-strict, i.e., we allow
the relaxed inequality $U(x) \ge S[\oboth](x, c)$.
This leads to the replacement of Rule~2 by a slightly modified rule:
\begin{itemize}
\item {\em Rule 2'.}
If there are candidates $x,y\in X \setminus \{c\}$ such that $U(x) > U(y)$ and $\ononmanip(y,x) - \wsummanip > U(y)$, set 
$U(x) \leftarrow U(y)$.
\end{itemize}
With everything else unchanged, we can easily re-prove a slightly different second condition for a successful manipulation (in the spirit of Theorem \ref{thm:uwwcm}):

\begin{theorem}\label{thm:cwwcm}
Let $(X,(\pi_1,\dots, \pi_k;w_1,\dots,w_k),$\linebreak$ (w_{k+1},\dots, w_n), c)$
be a \CWWCM\ instance. There exists a (finite) function $U:X\to\mathbb{Z}$ which can be computed in polynomial time such that the following two statements are equivalent:
\begin{enumerate}
\item There exist votes $(\pi_{k+1},\dots, \pi_n)$ such that $c$ is a Schulze winner for the weighted voting profile 
	$(\pi_1, \dots, \pi_n; w_1, \dots, w_n)$.
\item $U(x)\ge \omega_{(\pi_1,\dots, \pi_k;w_1,\dots, w_k)}(x,c)-(w_{k+1}+\dots+w_n)$ for all $x\in X\setminus \{c\}$.
\end{enumerate}
\end{theorem}

$\CWWCM$ is therefore solvable in polynomial time, too.

\section{Conclusion}

We studied weighted coalitional manipulation under the Schulze voting rule, and showed that constructive coalitional manipulation is polynomial-time solvable and a successfully manipulating voting profile for the manipulators is polynomial-time constructible in the unique and nonunique-winner models. Together with an adaption of the algorithm by Parkes and Xia for destructive manipulation~\cite{parkes-xia-2012} to weighted voting, this resolves the open questions in~\cite{hemaspaandra-lavaee-menton-2016} on the complexity of weighted coalitional manipulation for Schulze voting.

\section*{Acknowledgement}
We thank the anonymous reviewers for their valuable comments that helped improve the presentation of this article.



\end{document}